\documentclass[aps,prl,twocolumn,superscriptaddress,longbibliography,reprint,floatfix]{revtex4-2}

\usepackage[T1]{fontenc}
\usepackage{amsmath}
\usepackage{amssymb}
\usepackage{amsfonts}
\usepackage{amsthm}
\usepackage{mathtools}
\usepackage{bm}
\usepackage{graphicx}
\usepackage[dvipsnames]{xcolor}

\definecolor{myrefcolor}{rgb}{0.067,0.5,0.5}

\newtheorem{theorem}{Theorem}
\newtheorem{corollary}{Corollary}
\newtheorem{lemma}{Lemma}

\usepackage{times}

\newcommand{\tr}{\operatorname{Tr}}
\newcommand{\SP}{\mathrm{SP}}
\newcommand{\CSP}{\mathrm{CSP}}
\newcommand{\ket}[1]{\lvert #1\rangle}
\newcommand{\bra}[1]{\langle #1\rvert}
\newcommand{\kb}[2]{\lvert #1\rangle\!\langle #2\rvert}
\newcommand{\D}{\mathcal D}
\newcommand{\R}{\mathcal R}
\newcommand{\N}{\mathcal N}
\newcommand{\Jd}{\mathcal J_{\downarrow}}
\newcommand{\Ju}{\mathcal J_{\uparrow}}

\usepackage[
  breaklinks=true,
  colorlinks=true,
  linkcolor=myrefcolor,
  citecolor=myrefcolor,
  urlcolor=myrefcolor
]{hyperref}

\hypersetup{
  pdftitle={Environmental records unlock universal quantum computation from thermal decoherence},
  pdfauthor={Chenfeng Cao and Qi Zhao}
}

\begin{document}

\title{Environmental records unlock universal quantum computation from thermal decoherence}

\author{Chenfeng Cao}
\email{caocf@hku.hk}
\affiliation{HK Institute of Quantum Science \& Technology, The University of Hong Kong, Pokfulam Road, Hong Kong SAR, China}

\author{Qi Zhao}
\email{zhaoqi@cs.hku.hk}
\affiliation{QICI Quantum Information and Computation Initiative, School of Computing and Data Science,
The University of Hong Kong, Pokfulam Road, Hong Kong SAR, China}

\begin{abstract}
At a fixed thermal exposure, the same stabilizer processor can be classically simulable or quantum universal, depending on which environmental records its controller retains. We give an exact computational classification of energy-counting thermal-idle instruments in a quantum processor with ideal stabilizer control and independent local Markov baths. The relaxation time $T_1$, homogeneous coherence time $T_2$, and equilibrium excited-state population $p_e$  determine an exact computational boundary at $(1-p_e)T_2/T_1=1$. If every location lies at or below it, branchwise nonnegative stabilizer decompositions give an explicit efficient classical sampler for the adaptive circuit and its full time-resolved exchange record. Above it at a single repeatedly accessible location, a suitable idle duration and no-exchange conditioning supply distillable ancillas and enable universal quantum computation with polynomial overhead. At finite temperature on the resource side, erasing the record at a sufficiently long, unsplit exposure makes the averaged channel stabilizer measure-and-prepare, and even a terminal parity check then yields only simulable branches. Yet at that same exposure, retaining only the bit recording whether any exchange occurred still heralds distillable ancillas, because a thermal round trip restores the parity after its first exchange has already removed the coherence.
\end{abstract}

\maketitle

Clifford gates, stabilizer preparations, Pauli measurements, and
classical feedforward admit efficient classical
simulation~\cite{Gottesman1999,AaronsonGottesman2004}.
We ask when coupling such a stabilizer processor to a monitored
thermal bath preserves this tractability and when it supplies
resources for universal quantum computation. The answer depends on
both the thermal dynamics and the environmental record available to
the controller.

During an idle, a qubit may exchange energy with its bath, leaving
records in photons or material excitations.
Quantum jumps have been monitored in superconducting
circuits~\cite{Vijay2011,Minev2019}, while erasure conversion
exposes related events in atomic and superconducting
platforms~\cite{Wu2022,Kubica2023,Ma2023,Scholl2023,Levine2024}.
Under ideal energy counting that resolves emission and absorption,
each detected exchange is an energy-basis measurement followed by
a reset, hence a stabilizer operation. The quiet branch, in which
no exchange occurs, can instead generate nonstabilizer states from
stabilizer inputs.

Earlier work on noise-induced simulability thresholds asks when unrecorded decoherence
renders an initially universal or hard-to-sample circuit
tractable~\cite{vanDamHoward2009,FujiiTamate2016,TrivediCirac2022}.
A positive Clifford-and-reset decomposition of the record-averaged
thermal-relaxation channel has also been used for error-correction
simulation~\cite{Garner2026}. Environmental measurements have been
used to reverse noisy channels~\cite{GregorattiWerner2003,BuscemiChiribellaDAriano2005},
while engineered dissipative dynamics can implement computational
operations or prepare magic states~\cite{SantosEtAl2012,MartinezAzcona2025}.
Here the controlled circuit is tractable before coupling to the bath,
and we ask what the retained exchange record adds for a specified
energy-counting instrument.

At the level of resource preparation, amplitude damping can generate
nonstabilizer states through no-click postselection from a Bell
pair~\cite{TriguerosGuzman2026}, while thermal-operation criteria
relate magic generation and distillability to
temperature~\cite{deOliveiraJunior2026}.
For record-averaged phase-covariant damping, Ref.~\cite{Cao2026}
identified finite-temperature magic islands and the zero-temperature
rebirth condition $T_2>T_1$, and showed how parity-syndrome extraction
makes the resource accessible. These state-level results leave open the computational status of the full monitored process itself.

Erasing the exchange record gives the averaged thermal channel. Retaining it makes the conditional branches available to the controller. For independent local Markov baths with ideal stabilizer
control, we establish the exact boundary $(1-p_e)T_2/T_1=1$,
where $p_e$ is the equilibrium excited-state population.
If every location lies at or below this boundary, nonnegative
stabilizer decompositions of all recorded branches remain valid
on entangled stabilizer inputs and under adaptive reuse, giving
an efficient sampler for circuit outputs and full time-resolved
exchange records.
Above the boundary, a single repeatedly accessible location with
selectable exposure supplies distillable ancillas through its quiet
branch, enabling universal quantum computation with polynomial
overhead.
On the resource side at nonzero temperature, a single sufficiently
long unsplit exposure makes the averaged channel stabilizer
measure-and-prepare and every terminal-parity branch simulable,
while retaining only the exchange bit, which records whether any exchange occurred, still heralds distillable
ancillas at that same exposure.

\smallskip
\paragraph*{Independent local baths.}
All noise in our model sits in the monitored idles: the stabilizer operations themselves are ideal, as in the ideal-Clifford, noisy-ancilla abstraction of
fault-tolerant computation~\cite{BravyiKitaev2005}.  Idles may be imposed or selected, and Theorem~\ref{thm:closed}(a) covers either.  Part (b) selects the exposure at one favorable location, reuses it, and routes the prepared ancillas through the ideal stabilizer network.  We use a calibrated rotating frame with the computational basis
aligned to the energy basis and the known Hamiltonian removed. The ideal energy-counting monitor records every emission and absorption event, including its time and direction, while pure dephasing is unmonitored. The experimentally accessible record and its measurement backaction are platform dependent. Detector loss is treated separately below.

For each location we use a time-homogeneous GKSL thermal generator~\cite{GoriniKossakowskiSudarshan1976,Lindblad1976}.  At location $i$, $p_e^{(i)}$ is the equilibrium excited-state population and $T_1^{(i)}$ is the relaxation time.  During an unsplit idle, coherence decays exponentially with time constant $T_2^{(i)}$.  We use the homogeneous coherence time, excluding unresolved quasistatic broadening.  At each location we take $0\leq p_e^{(i)}\leq1/2$, fixing the downward rate
as the larger one. An inverted bath exchanges the energy labels. Rates and exposure times may vary across the processor.  Each location carries the calibration number
\begin{equation}
 \chi_i\equiv(1-p_{e}^{(i)})\frac{T_{2}^{(i)}}{T_{1}^{(i)}}-1,
 \qquad
 \chi_{\max}\equiv\max_i\chi_i,
 \label{eq:chi}
\end{equation}
and for a single location we drop the index.  We call $\chi_i\leq0$ the simulable side and $\chi_i>0$ the resource side.  With the emission rate $\Gamma_\downarrow=(1-p_e)/T_1$ and the absorption
rate $\Gamma_\uparrow=p_e/T_1$, let $\Gamma_\phi$ denote the homogeneous pure-dephasing rate, so that $1/T_2=(\Gamma_\uparrow+\Gamma_\downarrow)/2+\Gamma_\phi$.  The condition $\chi>0$ reads $1/T_2<\Gamma_\downarrow$, and its boundary is equivalently $\Gamma_\phi=(\Gamma_\downarrow-\Gamma_\uparrow)/2$.  The condition requires the coherence factor to decay more slowly than the excited-state no-exchange survival weight.  Combining $\chi>0$ with the complete-positivity limit $T_2\leq2T_1$ gives the physical resource region $T_1/(1-p_e)<T_2\leq2T_1$, which is nonempty for $p_e<1/2$.

\begin{figure}[!t]
\centering
\includegraphics[width=1\columnwidth]{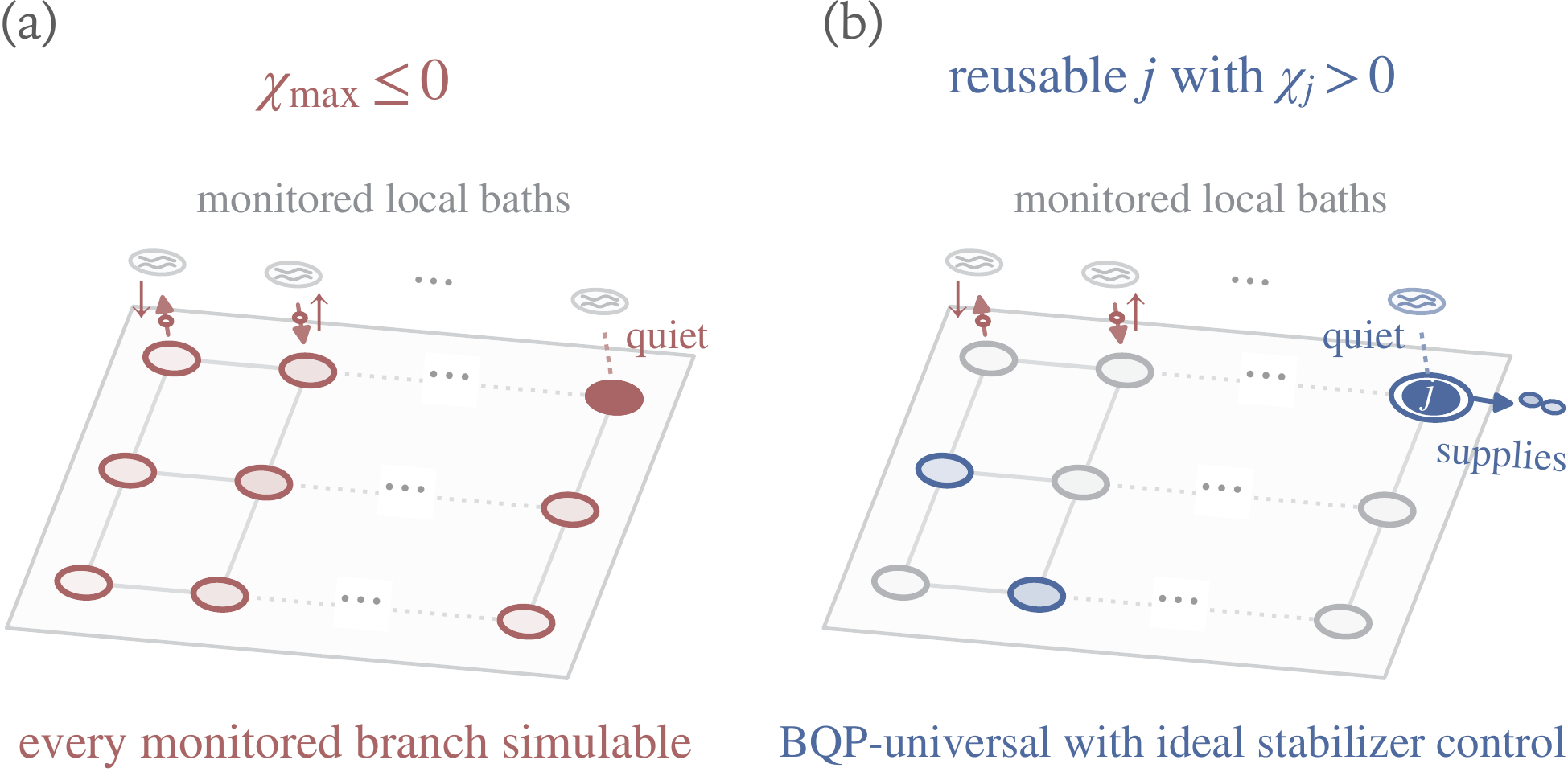}
\caption{The local calibrations determine the two sides of
Theorem~\ref{thm:closed}.  (a) Each thermal-idle location has its own $\chi_i$. If $\chi_{\max}\leq0$, every monitored call has a nonnegative stabilizer decomposition.  (b) If one repeatedly accessible location has $\chi_i>0$, its quiet branch provides distillable ancillas to the stabilizer processor.}
\label{fig:boundary}
\end{figure}

\begin{theorem}[Computational classification]\label{thm:closed}
Under the independent-bath model above, consider adaptive stabilizer circuits with polynomial-time classical control.
\par\noindent
(a) If $\chi_{\max}\leq0$, every record branch is a nonnegative sum of trace-nonincreasing stabilizer operations. Circuit outputs and full records admit a sampler exact in real arithmetic, with expected cost polynomial in the circuit size and in a uniform upper bound on the total monitored exposure in local-$T_2$ units.
\par\noindent
(b) A repeatedly accessible location with fixed calibrated parameters, $\chi>0$, and selectable exposure enables BQP-universal computation via quiet-branch preparation, with overhead polynomial in the circuit size and in $1/\chi$.
\end{theorem}

When every $\chi_i\leq0$, the local decompositions give a classical sampler. A single repeatedly accessible location with $\chi_i>0$ can supply ancillas to the entire processor, so the classification depends on $\chi_{\max}$ [Fig.~\ref{fig:boundary}].  A classical polynomial-time sampler for the entire resource-side family at sufficiently small fixed total-variation error would therefore imply $\mathrm{BPP}=\mathrm{BQP}$.  In the standard energy-jump representation of the thermal master equation~\cite{DalibardCastinMolmer1992,PlenioKnight1998,Daley2014}, the quiet branch decomposes as
\begin{equation}
\begin{aligned}
 \mathcal N_t={}&e^{-t/T_2}\mathcal I
 +\bigl(e^{-\Gamma_\downarrow t}-e^{-t/T_2}\bigr)\Delta_Z\\
 &+\bigl(e^{-\Gamma_\uparrow t}-e^{-\Gamma_\downarrow t}\bigr)
  \mathcal P_{0\to0},
\end{aligned}
 \label{eq:quiet}
\end{equation}
where $\mathcal I$ is the identity channel, $\Delta_Z$ is complete dephasing, and $\mathcal P_{i\to j}(\rho)=\langle i|\rho|i\rangle\kb jj$.  All three coefficients are nonnegative at every $t>0$ if and only if the three exponents are ordered, $\Gamma_\uparrow\leq\Gamma_\downarrow\leq1/T_2$.  The first inequality is $p_e\leq1/2$ and the second is $\chi\leq0$.  Detected jumps are computational-basis measurements followed by resets, so Eq.~\eqref{eq:quiet} gives a nonnegative decomposition of every record branch.

The local decompositions tensor across simultaneous exposures and compose under adaptive reuse, including on entangled inputs. Continuous-time uniformization then gives an event-driven tableau sampler: on each exposed rail, two auxiliary stabilizer updates compete with the two recorded resets, and their nonnegative rates sum to $1/T_2^{(i)}$ independently of the state, so each rail carries a fixed exponential clock and rates add across rails. The tableau supplies the branching probabilities and updates, and the resulting process reproduces the joint law of circuit outputs and exchange records. The Supplemental Material (SM) gives the construction and the finite-precision guarantee for binned records.

\begin{figure*}[t!]
\centering
\includegraphics[width=\textwidth]{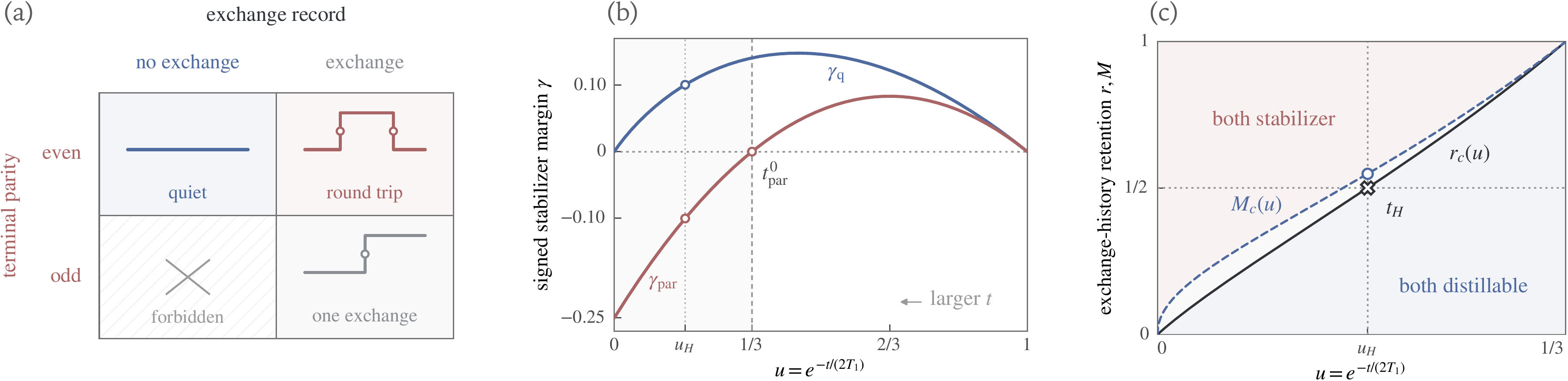}
\caption{(a) The exchange bit and the terminal-parity bit are incomparable: a history with even terminal parity can still contain exchanges, whereas odd parity implies at least one. (b) Signed stabilizer margins at $p_e=1/4$ and $T_2=2T_1$ versus $u=e^{-t/(2T_1)}$, with exposure increasing to the left.  The quiet margin $\gamma_{\rm q}$ stays positive at every finite exposure, while the terminal-parity margin $\gamma_{\rm par}$ crosses zero at $u=1/3$ ($t_{\rm par}^{0}=2T_1\ln3$).  Shading marks the window of
Theorem~\ref{thm:record}.  At $u_H=3-2\sqrt2$ the margins are equal and opposite. (c) Retention thresholds $r_c(u)$ [Eq.~\eqref{eq:rc}, solid] for controlled
acceptance and $M_c(u)$ (SM, dashed) for independent detector loss, with $M=m_\uparrow m_\downarrow$, on a common vertical axis. Below both curves the accepted branch is distillable under either rule, above both it has a nonnegative stabilizer decomposition, and in the gray region between them the two rules disagree. For $1/3<u<1$ the branch is distillable for every $r$ and every $M$, so both curves end at $u=1/3$.}
\label{fig:ladder}
\end{figure*}

\smallskip
\paragraph*{Reaching the resource side.}
To show that the boundary is tight, prepare $\ket+$ at a repeatedly accessible location and retain the runs with no energy exchange.  Because every other history is a reset branch, this preparation needs only the exchange bit.
The one-qubit stabilizer polytope is an octahedron in the Bloch ball.  The accepted output is nonstabilizer when it lies outside this octahedron.  We call the success-weighted signed excess over the relevant stabilizer facet the signed stabilizer margin $\gamma$.  Its positive part $[\gamma]_+$ is the success-weighted robustness surplus~\cite{HowardCampbell2017}.  For the quiet branch with $\gamma\leq0$, $-\gamma$ is the $\Delta_Z$ coefficient in Eq.~\eqref{eq:quiet}.

\begin{lemma}\label{lem:quietboundary}
For a monitored thermal idle of any duration $t>0$, the quiet branch applied to $\ket+$ succeeds with probability
$(e^{-\Gamma_\uparrow t}+e^{-\Gamma_\downarrow t})/2$ and has signed stabilizer margin
\begin{equation}
 \gamma_{\rm q}(t)=e^{-t/T_2}-e^{-\Gamma_\downarrow t},
 \label{eq:exchangegain}
\end{equation}
which is positive if and only if $\chi>0$.
\end{lemma}

The unnormalized quiet output has diagonal weights
$e^{-\Gamma_\uparrow t}/2$ and $e^{-\Gamma_\downarrow t}/2$ and coherence $e^{-t/T_2}/2$.  Its nonzero Bloch components are nonnegative, giving the success-weighted facet excess in Eq.~\eqref{eq:exchangegain}.  On the resource side the coherence outlives the excited-state weight, and the output crosses an octahedron facet.  An $\{I,H\}$ twirl moves it onto the Hadamard axis, where
the distillation threshold meets that facet~\cite{Reichardt2005}.  The optimal quiet exposure is $t_{\rm q}^{\star}=T_2\ln(1+\chi)/\chi$, with margin
$\gamma_{\rm q}^{\star}=\chi(1+\chi)^{-1-1/\chi}$.  Near the boundary this margin approaches $\chi/e$, while the success probability at $t_{\rm q}^{\star}$ stays above $1/e$ across the physical parameter range.  The shrinking margin, rather than heralding, sets the overhead.

Above the boundary, standard distillation prepares an $H$-type ancilla a fixed distance outside the stabilizer octahedron with $(1/\chi)^{O(1)}$ monitored-idle calls.  Below the boundary, no adaptive reuse of the allowed operations can prepare such a state.  Injection completes the BQP construction with the polynomial overhead stated in Theorem~\ref{thm:closed}(b)~\cite{BravyiKitaev2005,GottesmanChuang1999}.  Keeping the exchange record prevents the quiet component from being averaged with reset branches.

\smallskip
\paragraph*{Erasing the record.}

To isolate the role of the exchange record, we sum over its outcomes and retain only the record-averaged channel $\Phi_t$.  A channel is stabilizer measure-and-prepare when its normalized Choi state is a convex mixture of
products of one-qubit stabilizer states.  Equivalently, it can be realized by a stabilizer POVM followed by stabilizer preparations.  An explicit affine decomposition is given in Eq.~\eqref{em:mp} in End Matter.

\begin{theorem}[Stabilizer measure-and-prepare threshold]\label{thm:mp}
For $p_e>0$, the equation $e^{-t/T_2}=p_e(1-e^{-t/T_1})$ has a unique positive solution $t_{\rm mp}$. The record-averaged channel $\Phi_t$ is stabilizer measure-and-prepare if and only if $t\geq t_{\rm mp}$.  At every such exposure, $\Phi_t^{\otimes n}$ maps every input to a mixture of product stabilizer states.
\end{theorem}

Primitive quantum Markov semigroups eventually become entanglement breaking~\cite{HansonRouzeFranca2020}.  The present decomposition also removes nonstabilizerness from arbitrary many-qubit inputs, a conclusion that does not follow from one-qubit magic breaking~\cite{Patra2024}.  The stabilizer measure-and-prepare, entanglement-breaking, and magic-breaking thresholds are in general distinct.  At the parameters of Fig.~\ref{fig:ladder}, there is an exposure window in which the channel is both magic breaking and entanglement breaking yet not completely stabilizer preserving.  It produces no magic from any single-qubit input, but acting on one rail of an ebit it leaves a separable state outside the stabilizer polytope (SM).

Stabilizer measure-and-prepare channels are in particular completely stabilizer preserving, so on the resource side, for $t\geq t_{\rm mp}$, repeated uses of $\Phi_t$ remain classically simulable within an adaptive stabilizer processor when the record is erased.  At the same exposure, retaining the exchange bit and conditioning on no exchange yields distillable states with success probability and margin independent of processor width.  Because $t_{\rm mp}$ is finite only for $p_e>0$, this contrast occurs only at finite temperature.  At small $p_e$, however, it appears only at late exposures where the quiet margin is already well below its optimum.

\smallskip

\paragraph*{Comparing records.}

We compare two one-bit readouts of the same monitored experiment: the exchange bit and terminal parity. Retaining only the exchange bit preserves the one-use stabilizer-assisted yield available from the full time-resolved record, because every resolved jump branch is diagonal (End Matter). As a second one-bit readout, we record the terminal $ZZ$ parity of a Bell pair after exposing one rail to the thermal idle, and Clifford-decode the even sector. This terminal-parity readout has margin
\begin{equation}
 \gamma_{\rm par}(t)=e^{-t/T_2}-p_e
 -(1-p_e)e^{-t/T_1}.
 \label{eq:endgain}
\end{equation}
Both readouts retain one bit, but not the same information, as shown in Fig.~\ref{fig:ladder}(a).  A completed thermal round trip restores the terminal parity after its first jump has removed the quiet-branch coherence. Terminal parity accepts such a history, whereas the exchange bit rejects it.  Conversely, a history containing an exchange may end in either parity sector.  For $p_e>0$, neither bit can be obtained from the other by classical postprocessing~\cite{LeppajarviSedlak2021}.

Both bits can be acquired in one experiment and then discarded selectively.  Conditioning on no exchange gives the margin $\gamma_{\rm q}$. If the exchange record is erased while terminal parity is retained, the margin becomes $\gamma_{\rm par}$. Decoding both parity sectors and erasing both bits leaves the same state as the record-averaged bare $\ket+$ output, with margin
\begin{equation}
\gamma_{\varnothing}(t)
=
e^{-t/T_2}
-2p_e
-(1-2p_e)e^{-t/T_1}.
\end{equation}
These margins obey
$\gamma_{\rm q}\geq\gamma_{\rm par}\geq\gamma_{\varnothing}$ (SM).

The bare-output margin has initial slope $\chi_{\varnothing}/T_2$ and becomes positive at some exposure if and only if $\chi_{\varnothing}>0$, where
\begin{equation}
\chi_{\varnothing}
\equiv
(1-2p_e)\frac{T_2}{T_1}-1
=
\left.\chi\right|_{p_e\to2p_e}.
\end{equation}
In contrast, both one-bit readouts admit a positive margin for some exposure if and only if $\chi>0$, so they share the existence boundary $\chi=0$, which at finite temperature is less restrictive than the bare-probe boundary $\chi_{\varnothing}=0$.  With an ebit probe and a free terminal-parity measurement, the same record-averaged channel has maximum one-use yield $[\gamma_{\rm par}]_+$ and retains the $\chi=0$ boundary for selectable exposures (End Matter).

\begin{theorem}[Record dependence at fixed exposure]\label{thm:record}
Prepare the Bell-pair probe described above and let one rail undergo a single
thermal idle that is not subdivided.  Let $0<p_e<1/2$ and $\chi>0$, and let
$t_{\rm par}^{0}$ denote the unique positive root of
$\gamma_{\rm par}(t)=0$.  For every fixed finite exposure $t\geq t_{\rm par}^{0}$,
\begin{equation}
 \gamma_{\rm q}(t)>0\geq\gamma_{\rm par}(t),
 \label{eq:recordswitch}
\end{equation}
and every branch of the terminal-parity instrument then has a nonnegative stabilizer decomposition.  Its repeated adaptive use admits efficient classical sampling, whereas the quiet branch of the monitored exposure heralds a distillable state.
\end{theorem}

In the limit of long exposure, the record-averaged channel prepares the thermal state, and the terminal-parity probe gives $\gamma_{\rm par}\to-p_e$.  Conditioning on no exchange removes the reset histories, so $\gamma_{\rm q}(t)>0$ for all $0<t<\infty$ when $\chi>0$ [Fig.~\ref{fig:ladder}(b)].  The quantity $e^{-t/T_2}-p_e(1-e^{-t/T_1})$ of Theorem~\ref{thm:mp} equals $\gamma_{\rm par}(t)+e^{-t/T_1}$, which is $e^{-t_{\rm par}^{0}/T_1}>0$ at $t=t_{\rm par}^{0}$, so the measure-and-prepare threshold lies later, $t_{\rm mp}>t_{\rm par}^{0}$.

On the relaxation-limited line $T_2=2T_1$,
$t_{\rm par}^{0}=2T_1\ln[(1-p_e)/p_e]$.
As $p_e\to0$, the separation moves to late exposures and its quiet margin vanishes linearly in $p_e$, whereas the optimal preparation time approaches $2T_1\ln2$ (SM). Terminal parity requires a single $ZZ$ measurement after the idle. The exchange bit requires monitoring throughout, although no event times or directions need be stored.

A controlled acceptance rule within the even sector interpolates continuously between the two readouts.

\begin{corollary}\label{cor:acceptance}
In the setting of Theorem~\ref{thm:record}, use the same Bell-pair probe and retain the even terminal-parity sector.  Within that sector, accept every zero-exchange history and independently accept each history containing an
exchange with probability $0\leq r\leq1$.  The accepted branch heralds a distillable state for all $r$ below the threshold
\begin{equation}
 r_c(t)=\frac{\gamma_{\rm q}(t)}
 {\gamma_{\rm q}(t)-\gamma_{\rm par}(t)}.
 \label{eq:rc}
\end{equation}
At or above $r_c$, the accepted branch has a nonnegative stabilizer decomposition.  At $p_e=1/4$, $T_2=2T_1$, and
$t=t_H\equiv4T_1\ln(1+\sqrt2)$, the two margins are equal and opposite, so $r_c=1/2$.
\end{corollary}

We show $r_c$ as a function of exposure at these
parameters in Fig.~\ref{fig:ladder}(c).  At $p_e=0$, all
three margins coincide at $e^{-t/T_2}-e^{-t/T_1}$ and share the boundary $T_2=T_1$.  For finite $p_e$ and $\chi\ll p_e$, the optimized terminal-parity margin is quadratic in $\chi$, whereas the quiet margin opens linearly.

\smallskip
\paragraph*{Detector loss.}
Imperfect monitoring gives the second interpolation in Fig.~\ref{fig:ladder}(c). A no-click record can contain missed exchanges, so \emph{quiet} still denotes a physical zero-exchange history.  Let $m_\uparrow,m_\downarrow$ be the miss probabilities for absorption and emission.

\begin{figure}[!t]
\centering
\includegraphics[width=0.95\columnwidth]{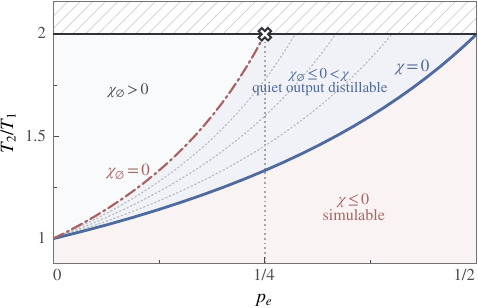}
\caption{Calibration plane for independent local baths.  The boundary uses
the homogeneous $T_2$ of the model.  The lines show the monitored boundary $\chi=0$, the bare-$\ket+$ short-time line $\chi_{\varnothing}=0$, and the complete-positivity limit $T_2=2T_1$.  The faint lines give the zero-initial-slope condition without terminal parity for absorption-miss
probabilities $m_\uparrow=1/4,1/2,3/4$.  The cross marks
$(p_e,T_2/T_1)=(1/4,2)$, the parameters of
Corollary~\ref{cor:acceptance}.}
\label{fig:map}
\end{figure}

With terminal parity, a no-click history with a single missed jump ends in the odd sector and is rejected.  The first accepted error is a complete round trip with both clicks missed, so independent losses enter only through $M=m_\uparrow m_\downarrow$ and only at second order in exposure.  The linear term remains $\chi/T_2$, and the boundary stays at $\chi=0$ for every $M$.  For $p_e>0$, larger $M$ lowers the margin and narrows
the useful exposure range.  If either direction is monitored perfectly, $M=0$ and the ideal quiet branch is recovered at every exposure.  Perfect emission monitoring with terminal parity suffices when absorption is unmonitored, since a no-click history then contains at most one absorption, which the even sector rejects. Erasure-conversion architectures provide related event flags~\cite{Kubica2023,Ma2023,Levine2024}; realizing this instrument also requires the corresponding measurement backaction.

At $p_e=1/4$, $T_2=2T_1$, and $t=t_H$, a positive margin requires $M<M_c\simeq0.548$.  Equal-efficiency detectors then need efficiency above $0.26$ in each direction.  At $50\%$ efficiency the margin is $0.057$, about $57\%$ of the ideal value.  Since only $M$ enters, the efficiencies need not be balanced.

Without terminal parity, one missed absorption changes the linear term.  For asymmetric loss, the initial slope is positive if and only if
\begin{equation}
 (1-m_\uparrow)\chi+m_\uparrow\chi_{\varnothing}
 =\left.\chi\right|_{p_e\to(1+m_\uparrow)p_e}>0.
 \label{eq:missedline}
\end{equation}
Missed absorption acts like additional excited-state population, moving the short-time line of the product $\ket+$ probe from $\chi=0$ at $m_\uparrow=0$ to $\chi_{\varnothing}=0$ at $m_\uparrow=1$ [Fig.~\ref{fig:map}].  At finite exposure the margin can also depend on $m_\downarrow$ and is nonincreasing in either miss probability. Finite-time dynamics, dark counts, and terminal-parity readout errors are treated in the SM.

\smallskip
\paragraph*{Calibration and controls.}
For scale, the $g$--$e$ transmon parameters reported in
Ref.~\cite{Elder2020} are $T_1\simeq51\,\mu\mathrm{s}$,
$T_2\simeq74\,\mu\mathrm{s}$, and $p_e\simeq0.4\%$.  The quoted $T_2$ was measured with a Ramsey sequence rather than an echo protocol (see Acknowledgments).  Using this free-induction time as a proxy in Eq.~\eqref{eq:chi} gives $\chi\simeq0.45$.  Within the model, the quiet branch is then optimized near $t_{\rm q}^{\star}\simeq1.20T_1$, succeeds in about two runs in three, and has margin $0.135$.  A second transmon calibration~\cite{Goldblatt2024}, using a Hahn-echo time as a proxy, gives $\chi\simeq0.007$ with a coherence-time uncertainty of about $0.015$ in $\chi$, leaving the sign unresolved (SM). The boundary thus lies within the calibration uncertainty of a current device.

The terminal-parity margin can be measured without full process tomography. The Bell-pair facet witness in the SM has expectation $1+2\gamma_{\rm par}$ and requires five two-qubit Pauli settings.  The same settings check phase-frame alignment.  Varying the acceptance probability
$r$ locates $r_c$.  A midpoint $X$ echo exchanges the survival weights and places both margins on or inside the stabilizer facet, giving a direct control experiment within the model (SM).

\smallskip
\paragraph*{Discussion and outlook.}
Theorem~\ref{thm:closed} classifies the complete monitored instrument: below the boundary every branch admits a nonnegative stabilizer decomposition, while above it monitoring separates the quiet branch, which supplies distillable states, from stabilizer reset histories. At finite temperature, thermal round trips make this distinction inaccessible to terminal parity alone. The resulting computational separation occurs at a fixed, unsplit exposure; allowing shorter idles restores the same existence boundary for the record-averaged channel with stabilizer ancillas.

The classification applies to instruments built from unsplit energy-counting idles in independent Markov baths with known energy axes.  Diffusive records~\cite{CampagneIbarcq2016}, feedback on continuous
signals~\cite{KarmakarEtAl2026}, intermediate parity checks, and alternative monitored unravelings define different instruments and require separate classifications~\cite{ChengIppoliti2023,BarchLidar2026}.  The locality of Eq.~\eqref{eq:chi} comes from bath factorization.  For two qubits sharing a common thermal bath (SM), any nonzero bath correlation opens a positive short-time margin even when both local calibration numbers obey $\chi_i\leq0$.  Local values of $T_1$, $T_2$, and $p_e$, defined from the diagonal rates and local dephasing, do not reveal this collective route, and no analogue of $\chi_{\max}$ is known for general correlated baths.

Related event records are already exposed in erasure architectures that flag relaxation and leakage~\cite{Wu2022,Kubica2023,Ma2023,Scholl2023,Levine2024}.  Below the local boundary such records can inform decoders.  Above it the exchange bit can herald states for injection, complementing fault-tolerant injection schemes that exploit erasure flags~\cite{Jacoby2025}. The accessible record and bath structure belong in the processor specification alongside $T_1$, $T_2$, $p_e$, and the reduced channel.

\smallskip
\smallskip
\paragraph*{Acknowledgments.}
\begin{acknowledgments}
We thank Christopher S. Wang for clarifying the Ramsey protocol used to measure the $T_2$ value reported in Table~S1 of Ref.~\cite{Elder2020}. C.C. was supported by the HKIQST Joint Post-doctoral Fellowship. Q.Z. acknowledges funding from Quantum Science and Technology-National Science and Technology Major Project 2024ZD0301900, National Natural Science Foundation of China (NSFC) via Project No. 12347104 and No. 12305030, Hong Kong Research Grant Council (RGC) via No. 27300823, 17310926, N\_HKU718/23, and R6010-23.
\end{acknowledgments}

\bibliographystyle{apsrev4-2}
\bibliography{refs}

\section*{End Matter}

\paragraph*{Phase-covariant decomposition.}
For the completely positive phase-covariant maps used here, write $\Phi:(x,y,z)\mapsto(\eta x,\eta y,bz+k)$ in the calibrated phase frame, with $\eta\geq0$ and $0\leq b\leq1$.  Conjugating both the input and output
by $X$ changes the sign of the displacement and preserves every stabilizer property, so we take $k\geq0$.  Here $\mathcal I,\mathcal X$, and $\mathcal Y$ denote conjugation by $I,X$, and $Y$.  A terminal-parity check
applies $\Phi$ to one rail of a stabilizer ebit, measures $ZZ$, and decodes the even sector.  This outcome occurs with probability $(1+b)/2$ and has logical Bloch vector $(2\eta,0,k)/(1+b)$.  Its signed stabilizer margin is
\begin{equation}
 \gamma_{\rm par}\equiv\frac{2\eta+k-1-b}{2},
 \label{em:generalparmargin}
\end{equation}
which reduces to Eq.~\eqref{eq:endgain} for the thermal channel.  Let $\mathcal P_{\to j}=\mathcal P_{0\to j}+\mathcal P_{1\to j}$ denote the unconditional reset $\mathcal P_{\to j}(\rho)=\tr(\rho)\kb jj$, so that
$\Delta_Z=\mathcal P_{0\to0}+\mathcal P_{1\to1}$.
When $\eta\geq b$ and $\gamma_{\rm par}\leq0$,
\begin{align}
 \Phi={}&\eta\mathcal I
 +\frac{\eta-b}{2}
 (\mathcal X+\mathcal Y)\nonumber\\
 &+(k-\gamma_{\rm par})\mathcal P_{\to0}
 -\gamma_{\rm par}\mathcal P_{\to1}.
 \label{em:channeldecomposition}
\end{align}
All weights are nonnegative.  As in Eq.~\eqref{eq:quiet}, the negative of the margin appears directly as a decomposition weight.  If $\eta<b$, the corresponding decomposition is
\begin{align*}
 \Phi={}&\eta\mathcal I
 +(b-\eta)\Delta_Z\nonumber\\
 &+\frac{1-b+k}{2}\mathcal P_{\to0}
 +\frac{1-b-k}{2}\mathcal P_{\to1}.
\end{align*}
Complete positivity makes all displayed coefficients nonnegative and forces $\gamma_{\rm par}<0$ in this case.  If $\gamma_{\rm par}>0$, the decoded terminal-parity state lies outside the octahedron.  Together with the two
decompositions, this proves that $\Phi$ is completely stabilizer preserving~\cite{SeddonCampbell2019,SaxenaGour2022,HeimendahlHeinrichGross2022} if and only if $\gamma_{\rm par}\leq0$ in the stated domain.  Complete-positivity details are given in the SM~\cite{Choi1975,RuskaiSzarekWerner2002}.

\smallskip
\paragraph*{Thermal instrument and quiet branch.}
For the thermal channel, write $\eta=e^{-t/T_2}$, $b=e^{-t/T_1}$, and $k=(1-2p_e)(1-b)$.  The no-exchange survival weights of $\ket0$ and $\ket1$ are $s_0=e^{-\Gamma_\uparrow t}$ and $s_1=e^{-\Gamma_\downarrow t}$.
Removing the detected recycling terms from the thermal master equation gives Eq.~\eqref{eq:quiet}, whose complete
positivity follows from $\eta^2\leq s_0s_1=b$.

On the line $T_2=2T_1$ and for $0\leq p_e<1/2$, the quiet output of Lemma~\ref{lem:quietboundary} is pure.  Its Bloch components are equal at
\begin{equation}
 t_H=\frac{2T_1\ln(1+\sqrt2)}{1-2p_e}.
 \label{em:hadamard}
\end{equation}
At this exposure the retained branch is the Hadamard state.

\smallskip
\paragraph*{Measure-and-prepare decomposition and witness.}
The record-averaged channel decomposes as
\begin{equation}
\begin{aligned}
 \Phi_t={}&\eta(\Delta_X+\Delta_Y)+b\Delta_Z
 +[(1-p_e)(1-b)-\eta]\mathcal P_{\to0}\\
 &+[p_e(1-b)-\eta]\mathcal P_{\to1}.
\end{aligned}
 \label{em:mp}
\end{equation}
Here $\Delta_P$ dephases in the $P$ basis.  The coefficients are nonnegative if and only if
$\eta\leq p_e(1-b)$.  Necessity follows from the product-stabilizer witness $W_{\rm mp}=XX-YY+IZ+ZZ-ZI$.  Every product stabilizer state obeys
$\langle W_{\rm mp}\rangle\leq1$.  On the normalized Choi state, $\langle XX\rangle=-\langle YY\rangle=\eta$,
$\langle IZ\rangle=k$, $\langle ZZ\rangle=b$, and
$\langle ZI\rangle=0$, so $\langle W_{\rm mp}\rangle=2\eta+k+b$.
Thus, the condition is necessary and sufficient.  Since the input-side marginal is $I/2$, the Choi decomposition corresponds to a stabilizer POVM followed by stabilizer preparations, proving Theorem~\ref{thm:mp}.

\smallskip
\paragraph*{Readout comparison and branch criterion.}
For the Bell-pair probe, and more generally after decoding an input supported on the even-parity code, the two terminal-parity outcomes induce
\begin{align}
 \mathcal E_{\rm even}&=\eta\mathcal I
 +(1-p_e+p_eb-\eta)\mathcal P_{0\to0}\nonumber\\
 &\quad+[p_e+(1-p_e)b-\eta]\mathcal P_{1\to1},\nonumber\\
 \mathcal E_{\rm odd}&=(1-p_e)(1-b)\mathcal P_{1\to0}
 +p_e(1-b)\mathcal P_{0\to1}.
 \label{em:parityprobe}
\end{align}
In the setting of Theorem~\ref{thm:record}, $\chi>0$ also gives $\eta>b$.
Thus, when $\gamma_{\rm par}\leq0$,
Eq.~\eqref{em:channeldecomposition} gives a nonnegative decomposition of the whole channel. The subsequent $ZZ$ measurement is a stabilizer operation, so each outcome retains a nonnegative branchwise decomposition.  This proves the simulable direction of Theorem~\ref{thm:record}.

Every classical coarse graining of the thermal trajectory has branches of the form
\begin{equation}
 \mathcal E_\ell(\rho)=
 \begin{pmatrix}
  A_\ell\rho_{00}+B_\ell\rho_{11}&c_\ell\rho_{01}\\
  c_\ell\rho_{10}&C_\ell\rho_{00}+D_\ell\rho_{11}
 \end{pmatrix},
 \label{em:branchform}
\end{equation}
with nonnegative coefficients.  For a fixed input, the contribution of an accepted branch is its success probability times the surplus $\mathcal R-1$ of the robustness of magic $\mathcal R$ of the normalized output~\cite{HowardCampbell2017}.  These contributions add over accepted branches.

\begin{lemma}[Branch criterion and one-use optimum]\label{lem:branch}
The branch has signed stabilizer margin $\gamma_\ell=c_\ell-\min\{A_\ell,D_\ell\}$ and a nonnegative stabilizer decomposition if and only if $\gamma_\ell\leq0$.  Its largest one-use, stabilizer-assisted yield is $[\gamma_\ell]_+$, attained by one Bell pair and one terminal-parity check.
\end{lemma}

\begin{proof}
When $\gamma_\ell\leq0$, the branch equals
\begin{align*}
 \mathcal E_\ell={}&c_\ell\mathcal I
 +(A_\ell-c_\ell)\mathcal P_{0\to0}
 +(D_\ell-c_\ell)\mathcal P_{1\to1}\\
 &+B_\ell\mathcal P_{1\to0}+C_\ell\mathcal P_{0\to1}.
\end{align*}
A product $Z$ eigenstate remains diagonal.  An equatorial Pauli eigenstate
has surplus
\begin{equation}
 [c_\ell-\min\{A_\ell+B_\ell,C_\ell+D_\ell\}]_+
 \leq[\gamma_\ell]_+.
\end{equation}
For an entangled stabilizer input, the stabilizer Schmidt normal form reduces
the exposed rail and one logical reference to a Bell pair. The terminal-parity check decodes the even Choi block to
$\frac12\bigl(\begin{smallmatrix}A_\ell&c_\ell\\
c_\ell&D_\ell\end{smallmatrix}\bigr)$, whose yield is
$[\gamma_\ell]_+$.  Strong monotonicity of robustness shows that no stabilizer post-processing can do better
~\cite{Fattal2004,HowardCampbell2017}.  If $\gamma_\ell>0$, the same decoded
block lies outside the stabilizer octahedron, so the branch cannot have a nonnegative stabilizer decomposition.
\end{proof}

Coarse graining adds branch coefficients.  The minimum is superadditive, so
a merged branch has margin at most $\sum_\ell\gamma_\ell$ and yield at most
$[\sum_\ell\gamma_\ell]_+\leq\sum_\ell[\gamma_\ell]_+$.  It cannot increase the total yield.  Resolved histories containing a jump have $c_\ell=0$, so the exchange bit already attains the complete record's total yield
$[\eta-s_1]_+$.  Within the retained even terminal-parity sector, the acceptance model has margin $(1-r)\gamma_{\rm q}+r\gamma_{\rm par}$, which gives Eq.~\eqref{eq:rc}.


\clearpage
\onecolumngrid

\begin{center}
\textbf{\large Supplemental Material for\\
``Environmental records unlock universal quantum computation from thermal decoherence''}
\end{center}

\setcounter{equation}{0}
\setcounter{figure}{0}
\setcounter{table}{0}
\setcounter{section}{0}
\renewcommand{\theequation}{S\arabic{equation}}
\renewcommand{\thefigure}{S\arabic{figure}}
\renewcommand{\theHequation}{S\arabic{equation}}
\renewcommand{\theHfigure}{S\arabic{figure}}
\renewcommand{\thesection}{S\arabic{section}}
\renewcommand{\theHsection}{S\arabic{section}}
\renewcommand{\theHtable}{S\arabic{table}}
\setcounter{secnumdepth}{1}

\vspace{4pt}
\noindent
This Supplemental Material supplies the proofs and derivations used in the
Letter.  The general phase-covariant parameters $(\eta,b,k)$ follow the End
Matter convention.  For the thermal channel,
$\eta=e^{-t/T_2}$, $b=e^{-t/T_1}$, and $k=\tau(1-b)$, where
$\tau=1-2p_e$.  The no-exchange survival weights are
$s_0=e^{-\Gamma_\uparrow t}$ and $s_1=e^{-\Gamma_\downarrow t}$, with the
rates in Eq.~\eqref{sm:rates}.

\section{Phase-covariant channel and exact stabilizer region}
\label{sm:channel}

We use the End Matter convention $\eta,k\geq0$ and $0\leq b\leq1$.  In
Eq.~\eqref{sm:choimatrix}, the reference qubit is the first tensor factor.
The normalized Choi state is
\begin{equation}
 J_\Phi=\frac14
 \begin{pmatrix}
 1+b+k&0&0&2\eta\\
 0&1-b-k&0&0\\
 0&0&1-b+k&0\\
 2\eta&0&0&1+b-k
 \end{pmatrix}.
 \label{sm:choimatrix}
\end{equation}
Positivity of the odd entries and of the even $2\times2$ block requires
\begin{equation}
 k\leq1-b,\qquad \sqrt{4\eta^2+k^2}\leq1+b.
 \label{sm:cp}
\end{equation}
These conditions are necessary and sufficient for complete positivity.

The one-qubit stabilizer set is $|x|+|y|+|z|\leq1$.  We write $\SP$ for
stabilizer-preserving channels, which map this octahedron
into itself, and $\CSP$ for completely stabilizer-preserving channels, for
which the same property holds after adjoining a stabilizer reference.  Testing
the six vertices establishes
\begin{equation}
 \Phi\in\SP\quad\Longleftrightarrow\quad \eta+k\leq1.
 \label{sm:sp}
\end{equation}
The two decompositions and the decoded terminal-parity test in End Matter
establish the complete stabilizer-preserving criterion.  For later use, the
terminal-parity outcome probability and decoded Bloch vector are
\begin{equation}
 p_+=\frac{1+b}{2},\qquad
 \widetilde{\bm r}_+=\frac{(2\eta,0,k)}{1+b}.
 \label{sm:terminalstate}
\end{equation}
The decoded state is outside the octahedron if and only if
$\gamma_{\rm par}>0$.  Hence
\begin{equation}
 \Phi\in\CSP\quad\Longleftrightarrow\quad
 2\eta+k\leq1+b
 \quad\Longleftrightarrow\quad\gamma_{\rm par}\leq0.
 \label{sm:csp}
\end{equation}

A channel with a nonnegative stabilizer decomposition is sampled by drawing
one term from its state-independent weights and applying the corresponding
Pauli, dephasing, or reset tableau update.  Such decompositions are closed under tensor
products and adaptive stabilizer composition.  Section~\ref{sm:instrument}
constructs the continuous-time tableau sampler for the complete monitored
instrument~\cite{AaronsonGottesman2004}.

For a direct experimental test, define
\begin{equation}
 W=IZ+ZI+XX+XY+YX-YY-ZZ.
 \label{sm:witness}
\end{equation}
Every two-qubit stabilizer state obeys $\langle W\rangle\leq1$.  To see this analytically, organize the support into the three multiplicatively closed triples $\{IZ,ZI,ZZ\}$, $\{XX,YY,ZZ\}$, and $\{XY,YX,ZZ\}$, which share $ZZ$.  Any two nonshared operators from different triples anticommute, so a pure stabilizer group intersects the support in at most one complete triple or one isolated element.  The stabilizer product rule bounds each possible contribution by one, and convexity covers mixtures.  The 33 pure stabilizer states that saturate the bound span a 14-dimensional affine hull, so the inequality defines a facet of the two-qubit stabilizer polytope.  On the thermal Choi state, $\tr(WJ_{\Phi_t})=1+2\gamma_{\rm par}(t)$.  Thus five settings, $XX$, $XY$, $YX$, $YY$, and $ZZ$, read the facet, with the one-body terms obtained from the $ZZ$ shots.  These settings also give the alignment quadrature $\langle XY\rangle+\langle YX\rangle$, which must vanish in the calibrated phase-covariant frame.

\section{Energy-counting thermal instrument}
\label{sm:instrument}

We work in the calibrated rotating frame, with the known Hamiltonian removed.  Every exposure interval and active-rail schedule is a classically specified part of the circuit.

Let
\begin{equation}
 \Gamma_\uparrow=\frac{p_e}{T_1},\qquad
 \Gamma_\downarrow=\frac{1-p_e}{T_1},\qquad
 \frac1{T_2}=\frac{\Gamma_\uparrow+\Gamma_\downarrow}{2}+\Gamma_\phi.
 \label{sm:rates}
\end{equation}
We use the dissipator
$\D[L]\rho=L\rho L^\dagger-\frac12\{L^\dagger L,\rho\}$.
The thermal master equation is
\begin{equation}
 \dot\rho=\Gamma_\downarrow\D[\sigma_-]\rho
 +\Gamma_\uparrow\D[\sigma_+]\rho
 +\frac{\Gamma_\phi}{2}\D[Z]\rho,
 \label{sm:lindblad}
\end{equation}
where $\sigma_-=\kb{0}{1}$ and $\sigma_+=\kb{1}{0}$ are the
emission and absorption operators, respectively.
Resolving the two energy-jump channels but not the dephasing bath produces a record consisting of jump times and directions.  Between jumps, the unnormalized state obeys the no-count equation, whose solution is
\begin{equation}
 \N_t(\rho)=
 \begin{pmatrix}
 e^{-\Gamma_\uparrow t}\rho_{00}&e^{-t/T_2}\rho_{01}\\
 e^{-t/T_2}\rho_{10}&e^{-\Gamma_\downarrow t}\rho_{11}
 \end{pmatrix},
 \label{sm:nojump}
\end{equation}
which is the quiet branch of Eq.~\eqref{eq:quiet}.  For $t>0$, its determinant condition is
$e^{-2t/T_2}\leq e^{-t/T_1}$, equivalent to the physical bound $T_2\leq2T_1$.
The convention $p_e\leq1/2$ fixes the energy labels.  For $p_e>1/2$,
conjugation by $X$ exchanges the two levels and their reset records in this
section.

For $p_e\leq1/2$, $s_0\geq s_1$.  Moreover,
\begin{equation}
 \chi\leq0
 \quad\Longleftrightarrow\quad
 \frac1{T_2}\geq\frac{1-p_e}{T_1}
 \quad\Longleftrightarrow\quad \eta\leq s_1
 \quad (t>0).
 \label{sm:closediff}
\end{equation}
Matching matrix elements then proves the stabilizer decomposition in Eq.~\eqref{eq:quiet}.  The jump superoperators at a recorded time are
\begin{equation}
 \Jd(\rho)=\Gamma_\downarrow\rho_{11}\kb00,\qquad
 \Ju(\rho)=\Gamma_\uparrow\rho_{00}\kb11.
 \label{sm:jumpmaps}
\end{equation}
They are jump-rate densities,
$\Jd=\Gamma_\downarrow\mathcal P_{1\to0}$ and
$\Ju=\Gamma_\uparrow\mathcal P_{0\to1}$, proportional to
trace-nonincreasing stabilizer maps.  A click in $[t',t'+dt']$ contributes the
instrument element $\Jd\,dt'$ or $\Ju\,dt'$.  More explicitly, a quiet
segment followed by a jump has density
\begin{align}
 (\Jd\circ\N_{t'})(\rho)
 &=\Gamma_\downarrow e^{-\Gamma_\downarrow t'}\rho_{11}\kb00,
 \nonumber\\
 (\Ju\circ\N_{t'})(\rho)
 &=\Gamma_\uparrow e^{-\Gamma_\uparrow t'}\rho_{00}\kb11.
 \label{sm:eventbranches}
\end{align}
For an exposed rail $A$ entangled with a reference $R$, the factor $\rho_{11}\kb00$ in the first line is replaced by $\kb00_A\otimes\bra1\rho_{AR}\ket1_A$, and the upward line exchanges $0$ and $1$.  A jump is therefore a local $Z$ measurement, a conditional tableau update on $R$, and a computational-basis reset of $A$.

A Poisson uniformization makes the multiqubit sampler explicit.  Let $\mathcal L_0$ denote the generator conditioned on no recorded energy exchange.  On the simulable side it can be written as
\begin{equation}
 \mathcal L_0=-\frac{1}{T_2}\mathcal I
 +\left(\frac{1}{T_2}-\Gamma_\downarrow\right)\Delta_Z
 +(\Gamma_\downarrow-\Gamma_\uparrow)\mathcal P_{0\to0}.
 \label{sm:uniformization}
\end{equation}
Both nontrivial coefficients are nonnegative when $\chi\leq0$ and $p_e\leq1/2$.  On a sampled stabilizer trajectory, let $\pi_z$ be the population of the exposed rail in $\ket z$.  These populations belong to $\{0,1/2,1\}$ and are obtained directly from the global tableau.  The total trace rate of two auxiliary updates and the two recorded jumps is
\begin{equation}
 \left(\frac1{T_2}-\Gamma_\downarrow\right)
 +(\Gamma_\downarrow-\Gamma_\uparrow)\pi_0
 +\Gamma_\downarrow \pi_1+\Gamma_\uparrow \pi_0
 =\frac1{T_2}.
 \label{sm:uniformrate}
\end{equation}
The next proposal time is therefore exponential with state-independent rate $1/T_2$.  At a proposal, the four choice probabilities are the four rates on the left of Eq.~\eqref{sm:uniformrate}, each multiplied by $T_2$.  They correspond to an auxiliary $\Delta_Z$ update, an auxiliary projection onto $\ket0$, a recorded downward reset, and a recorded upward reset.  A $\Delta_Z$ update is sampled as a $Z$ measurement whose outcome remains internal to the simulator.  The auxiliary $\mathcal P_{0\to0}$ update is a $Z$ measurement conditioned on outcome $0$, and its proposal rate vanishes when $\pi_0=0$.  Every update is a tableau operation, including when the exposed rail is entangled with the rest of the processor.

For several simultaneously exposed rails, the proposal rates add.  The simulator draws from the total rate $\Lambda=\sum_i(1/T_{2}^{(i)})$ and chooses rail $i$ with probability $1/(T_{2}^{(i)}\Lambda)$.  Updating the global tableau after each event retains all correlations between the rails.  Summing the auxiliary events between recorded jumps recovers the Dyson expansion of $e^{t'\mathcal L_0}=\mathcal N_{t'}$.  The remaining marked events insert exactly the physical maps $\Jd$ and $\Ju$.  Induction over recorded events and intervening stabilizer operations proves equality of the complete joint law, including adaptive stopping rules and feedforward.

Let $H_{\max}$ bound, uniformly over detector records and adaptive choices,
the total monitored exposure in units of the local $T_2$, and suppose that
the classical controller runs in polynomial time.  The number of proposals is dominated by a Poisson variable of mean $H_{\max}$, and every proposal requires one polynomial-time tableau update.  The ideal real-number model therefore provides an exact stochastic representation with expected polynomial cost whenever $H_{\max}$ is polynomial.  For the finite-precision guarantee, the polynomial-time controller acts on detector records binned at the stated resolution and uses finitely represented control times.  Assuming the calibrated rates are efficiently computable, truncating the Poisson tail after $O(H_{\max}+\log(1/\varepsilon))$ proposals and evaluating cumulative probabilities to the corresponding precision keeps the total-variation error below $\varepsilon$, with overhead polynomial in $\log(1/\varepsilon)$ and the timestamp bit length.  Postselection on an exponentially rare outcome can still be inefficient.

On the resource side, apply Eq.~\eqref{sm:nojump} to $\kb++$.  Its trace is
$(s_0+s_1)/2$, and the normalized Bloch vector is
$(2\eta,0,s_0-s_1)/(s_0+s_1)$.  Because all coordinates are nonnegative,
\begin{equation}
 \frac{s_0+s_1}{2}
 \left(\frac{2\eta+s_0-s_1}{s_0+s_1}-1\right)
 =\eta-s_1.
 \label{sm:nojumpgain}
\end{equation}
Thus, for every chosen duration $t>0$, the branch leaves the stabilizer
octahedron if and only if $\chi>0$.  After the $\{I,H\}$ twirl, this is the distillable region of Ref.~\cite{Reichardt2005} (Sec.~\ref{sm:hierarchy}).

\section{Reciprocal common-bath block}
\label{sm:common}

Consider two identical qubits in a reciprocal common thermal
bath~\cite{Dicke1954,Lehmberg1970,MingantiEtAl2021}, whose exchange part is
\begin{align}
 \mathcal L_{\rm ex}(\rho)={}&
 \sum_{i,j=1}^{2}\Gamma^\downarrow_{ij}
 \left(\sigma_-^{(i)}\rho\sigma_+^{(j)}
 -\frac12\{\sigma_+^{(j)}\sigma_-^{(i)},\rho\}\right)
 +\sum_{i,j=1}^{2}\Gamma^\uparrow_{ij}
 \left(\sigma_+^{(i)}\rho\sigma_-^{(j)}
 -\frac12\{\sigma_-^{(j)}\sigma_+^{(i)},\rho\}\right),
 \label{sm:commonmaster}
\end{align}
with reciprocal KMS rate matrices
\begin{equation}
 \Gamma^\downarrow=\frac{1-p_e}{T_1}
 \begin{pmatrix}1&\kappa\\\kappa&1\end{pmatrix},
 \qquad
 \Gamma^\uparrow=\frac{p_e}{T_1}
 \begin{pmatrix}1&\kappa\\\kappa&1\end{pmatrix},
 \qquad |\kappa|\leq1.
 \label{sm:commonrates}
\end{equation}
Here $T_1$ and $p_e$ are defined by the diagonal local rates, with
$T_2^{-1}=(2T_1)^{-1}+\Gamma_\phi$; collective coupling need not leave
the reduced single-qubit dynamics exponential.
Local frequency shifts are absorbed into calibrated rotating frames, and
residual coherent exchange is neglected.  We include independent homogeneous
dephasing as $(\Gamma_\phi/2)\sum_i\D[Z_i]$.  Diagonalizing
Eq.~\eqref{sm:commonrates} produces symmetric and antisymmetric modes with initial
no-click hazards $(1\pm\kappa)/T_1$, so their splitting determines
$|\kappa|$.  On the excitation-hole pair, the upward and downward hazards
add to $1/T_1$, giving common attenuation, while $\kappa$ controls the relative
quiet filtering.

For the two-qubit sector $\mathcal H_1=\operatorname{span}\{\ket{10},\ket{01}\}$,
define $\ket{0_L}=\ket{10}$, $X_L=X_1X_2$, and $Z_L=Z_2$
(equivalently $-Z_1$ on $\mathcal H_1$).
Removing all monitored recycling terms leaves
\begin{equation}
 \dot\rho=-\frac{1}{2T_1}\{I+\kappa X_L,\rho\}
 +\Gamma_\phi\D[Z_L]\rho.
 \label{sm:commonquiet}
\end{equation}
The anticommutator filters the logical $X_L$ eigenmodes, while $\D[Z_L]$
removes the codeword coherence at rate $2\Gamma_\phi$.
Starting from $\ket{0_L}$, the
filter creates an $X_L$ component at order $t$, so the dephasing term first
changes its facet margin at order $t^2$.
When $\Gamma_\phi=0$, equivalently $T_2=2T_1$ in this
independent-dephasing model, the quiet Kraus operator is
\begin{equation}
 K_0\big|_{\mathcal H_1}=e^{-t/(2T_1)}
 \exp\!\left[-\frac{\kappa t}{2T_1}X_L\right].
 \label{sm:commonfilter}
\end{equation}
Applied to $\ket{0_L}$, the unnormalized state has trace
$e^{-t/T_1}\cosh(|\kappa|t/T_1)$, logical $X_L$ magnitude
$e^{-t/T_1}\sinh(|\kappa|t/T_1)$, and logical $Z_L$ component
$e^{-t/T_1}$.  Its facet excess after Clifford decoding is
$\gamma_{\rm coll}(t)=e^{-t/T_1}[1-e^{-|\kappa|t/T_1}]$.  For
$\kappa\neq0$, it is largest at
$t=T_1\ln(1+|\kappa|)/|\kappa|$, where it equals
$|\kappa|(1+|\kappa|)^{-1-1/|\kappa|}$.

For finite relative-phase dephasing, put
$\omega=\sqrt{\Gamma_\phi^2+\kappa^2/T_1^2}$.  The solution of
Eq.~\eqref{sm:commonquiet} is
\begin{equation}
 \gamma_{\rm coll}(t)=e^{-t/T_1}\left\{1-e^{-\Gamma_\phi t}
 \left[\cosh(\omega t)+
 \frac{\Gamma_\phi-|\kappa|/T_1}{\omega}\sinh(\omega t)\right]\right\}.
 \label{sm:commondephasing}
\end{equation}
At $\Gamma_\phi=\kappa=0$, the product
$(\Gamma_\phi-|\kappa|/T_1)\sinh(\omega t)/\omega$ is defined by continuity
as zero.  For every finite $\Gamma_\phi$, the expansion of
Eq.~\eqref{sm:commondephasing} at the origin is
\begin{equation}
 \gamma_{\rm coll}(t)=\frac{|\kappa|}{T_1}t+O(t^2).
 \label{sm:commonslope}
\end{equation}
Every $\kappa\neq0$ therefore admits a sufficiently short exposure with
positive margin for any finite $\Gamma_\phi$.  At $\kappa=0$, the quiet evolution on
$\mathcal H_1$ is attenuation followed by $Z_L$ dephasing.  With
$\pm$-mode resolution, the reachable set
$\{\ket{00},\ket{11}\}\cup\mathcal H_1$ is closed.  A click maps
$\mathcal H_1$ to $\ket{00}$ or $\ket{11}$, or maps either of those states
back to a logical $X_L$ eigenstate, while no-click evolution on $\ket{00}$ and
$\ket{11}$ is scalar.  Thus, every trajectory branch is stabilizer at
$\kappa=0$.  For $\kappa\neq0$, the click branches remain stabilizer and the
quiet branch on $\mathcal H_1$ has the positive short-time margin above.

The encoding is part of the model.  If a collective jump acts on an
unrestricted input, it need not preserve stabilizers.  For example,
$(\sigma_-^{(1)}+\sigma_-^{(2)})\ket{++}
\propto2\ket{00}+\ket{01}+\ket{10}$.  The monitored output basis is also part
of the instrument.  Even at $\kappa=0$, resolving a rotated combination
$\cos\theta\,\sigma_-^{(1)}+\sin\theta\,\sigma_-^{(2)}$ is admissible but has
a nonstabilizer logical effect unless $\theta\in(\pi/4)\mathbb Z$.

\section{Heterogeneous locations}
\label{sm:heterogeneity}

Let $i$ label factorized local thermal idle locations with separately monitored
records, known computational axes, and parameters $T_{1}^{(i)}$,
$T_{2}^{(i)}$, and $p_{e}^{(i)}$.  Define $\chi_i$ by Eq.~\eqref{eq:chi}
using those parameters.  For an exposure $t_i$, write
$\eta_i=e^{-t_i/T_2^{(i)}}$,
$s_{0,i}=e^{-p_e^{(i)}t_i/T_1^{(i)}}$, and
$s_{1,i}=e^{-(1-p_e^{(i)})t_i/T_1^{(i)}}$.  If $\chi_i\leq0$, then
the local coefficients obey
\begin{equation}
 \eta_i\leq s_{1,i}\leq s_{0,i}
 \qquad (t_i\geq0),
 \label{sm:localordering}
\end{equation}
and the quiet map has the nonnegative decomposition in Eq.~\eqref{eq:quiet}.  Every detected exchange is still a measurement and reset.  Tensor products of simultaneous branches and adaptive compositions of successive branches therefore remain stabilizer instruments with nonnegative weights.  The classical draws can depend on the location, exposure, and tableau-computable local populations, but on no other feature of the state.  Hence $\max_i\chi_i\leq0$ is the condition under which every callable branch has the nonnegative decomposition used by the sampler, and it suffices for classical simulation of an arbitrary array of independent local thermal instruments.

Conversely, suppose a repeatedly accessible monitored location $j$ has $\chi_j>0$.  Its quiet branch obeys $\eta_j>s_{1,j}$ at every finite exposure.  Applying that branch to $\ket+$ prepares the distillable state in Eq.~\eqref{sm:nojumpgain}.  Repeated preparation yields injection states under the selectable access and ideal stabilizer control assumed in the Letter.  The heterogeneous boundary is determined by the largest local $\chi_i$ rather than by averaged calibration parameters.

\section{Record refinement, timing, and the finite-temperature Hadamard point}
\label{sm:hierarchy}

The record-averaged product-probe output of $\ket+$, the bare output of the Letter, is
$\bm r_{\varnothing}=(\eta,0,\tau(1-b))$.  The three margins are
\begin{align}
 \gamma_{\varnothing}=\eta+\tau(1-b)-1
 &=\eta-[2p_e+(1-2p_e)b],
 \label{sm:averagedgain}
 \\
 \gamma_{\rm par}=p_+(\|\widetilde{\bm r}_+\|_1-1)
 &=\eta-[p_e+(1-p_e)b],
 \label{sm:endgain}
 \\
 \gamma_{\rm q}&=\eta-b^{1-p_e}.
 \label{sm:exchangegain}
\end{align}
The middle line follows from Eq.~\eqref{sm:terminalstate} with
$k=\tau(1-b)$, and the last from Eq.~\eqref{sm:nojumpgain}.  Subtracting in
pairs,
\begin{align*}
 \gamma_{\rm par}-\gamma_{\varnothing}&=p_e(1-b),\\
 \gamma_{\rm q}-\gamma_{\rm par}
 &=p_e\cdot1+(1-p_e)b-b^{1-p_e}.
\end{align*}
The second line is nonnegative by the weighted arithmetic--geometric mean
inequality with weights $p_e$ and $1-p_e$,
\begin{equation*}
 p_e\cdot1+(1-p_e)\,b\geq 1^{\,p_e}\,b^{\,1-p_e}=b^{1-p_e},
\end{equation*}
with equality only at $b=1$ or $p_e=0$.  Both differences vanish only at $p_e=0$ or
$t=0$.  At $p_e=0$, a single idle contains at most one exchange, so the
exchange and terminal-parity bits are equivalent and their margins coincide
at every exposure.

For $T_2=2T_1$ and $0<p_e<1/2$, set $u=e^{-t/(2T_1)}$.  Then the terminal-parity margin factors,
\begin{equation*}
 \gamma_{\rm par}=u-p_e-(1-p_e)u^2
 =(1-p_e)(1-u)\left(u-\frac{p_e}{1-p_e}\right),
\end{equation*}
so it is positive on $u_0<u<1$ and nowhere else, where $u_0=p_e/(1-p_e)$.
At that exposure,
$\gamma_{\rm q}=u_0-u_0^{2(1-p_e)}$, and the quiet-branch success probability
is $[u_0^{2p_e}+u_0^{2(1-p_e)}]/2$.  The margin vanishes linearly with
$p_e$, whereas the quiet-branch success probability tends to $1/2$ and
$t_{\rm par}^{0}\sim2T_1\ln(1/p_e)$.

These three levels can be nested in one acquisition.  Prepare a Bell pair,
monitor energy exchange on the exposed rail, and finally measure $ZZ$.  On
the no-exchange subset the state lies in the even sector, and Clifford
decoding recovers the quiet-branch state with margin $\gamma_{\rm q}$.  Erasing
the exchange record while retaining terminal parity leaves margin $\gamma_{\rm par}$.
The complementary sector has $p_-=(1-b)/2$ and
$\widetilde{\bm r}_-=(0,0,\tau)$.  If each sector is decoded and its label is
then erased, the weighted Bloch vectors sum to
\begin{equation}
 p_+\widetilde{\bm r}_+ + p_-\widetilde{\bm r}_-
 =\bigl(\eta,0,\tau(1-b)\bigr),
 \label{sm:nestedrecords}
\end{equation}
the record-averaged output used in Eq.~\eqref{sm:averagedgain}.  Thus, the inequalities compare genuine classical coarse grainings of a common experiment.

Since $(1-p_e)/T_1=(1+\chi)/T_2$, the quiet margin is
$\gamma_{\rm q}(t)=e^{-t/T_2}-e^{-(1+\chi)t/T_2}$.  On the resource side,
stationarity requires
$e^{-t/T_2}=(1+\chi)e^{-(1+\chi)t/T_2}$, that is $e^{\chi t/T_2}=1+\chi$,
so the unique stationary point is
\begin{equation}
 t_{\rm q}^{\star}=\frac{T_2\ln(1+\chi)}{\chi},\qquad
 \gamma_{\rm q}^{\star}
 =(1+\chi)^{-1/\chi}\left(1-\frac{1}{1+\chi}\right)
 =\chi(1+\chi)^{-1-1/\chi}.
 \label{sm:exchangeopt}
\end{equation}
Since $(1+\chi)^{-1/\chi}\to e^{-1}$, the margin opens as $\chi/e+O(\chi^2)$.
At this exposure, $s_1^\star=(1+\chi)^{-1-1/\chi}$ and
$s_0^\star=(s_1^\star)^{p_e/(1-p_e)}$.  Complete positivity gives $\chi\leq1-2p_e$ and hence
$p_e/(1-p_e)\leq(1-\chi)/(1+\chi)$.  The success probability
$p_{\rm s}^\star=(s_0^\star+s_1^\star)/2$ therefore obeys
\begin{equation}
 p_{\rm s}^\star\geq
 \frac{(1+\chi)^{1-1/\chi}+(1+\chi)^{-1-1/\chi}}{2}
 =(1+\chi)^{-1/\chi}\cosh[\ln(1+\chi)]>\frac{1}{e}.
 \label{sm:successbound}
\end{equation}
The bound tends to $1/e$ as $\chi\to0^+$.  Also, since $\chi\leq1$ and $s_1^\star$ decreases in $\chi$,
$s_1^\star\geq1/4$, so
$\gamma_{\rm q}^\star=\chi s_1^\star\geq\chi/4$.  After the
$\{I,H\}$ twirl, write the normalized Bloch coordinates as $r_x=r_z=x$.
Since $p_{\rm s}^\star\leq1$ and
$x-\frac12=\gamma_{\rm q}^\star/(2p_{\rm s}^\star)$, the margin bound implies
$x-\frac12\geq\chi/8$.  One Steane-code $H$-state round maps $x$ to
$f(x)=x^3(7+8x^4)/(1+14x^4)$ with acceptance probability
$(1+14x^4)/64$~\cite{Reichardt2005}.  At $x=1/2$, $f'(1/2)=7/5$ and the
acceptance probability is $15/512$.  Hence $O(\log(1/\chi))$ constant-cost
rounds reach a fixed excess using $(1/\chi)^{O(1)}$ raw states.  Subsequent
distillation and injection have the usual polynomial
overhead~\cite{BravyiKitaev2005}.

The exact terminal optimum is
\begin{equation}
 \frac{t_{\rm par}^{\star}}{T_1}
 =\frac{(1+\chi)\ln(1+\chi)}{\chi+p_e},\qquad
 \gamma_{\rm par}^{\star}
 =(\chi+p_e)(1+\chi)^{-(1+\chi)/(\chi+p_e)}-p_e.
 \label{sm:terminalexact}
\end{equation}
For fixed $0<p_e<1/2$, it opens quadratically.  Expanding in $t/T_2$ gives
$\gamma_{\rm par}=\chi t/T_2-\frac12
\frac{p_e+2\chi+\chi^{2}}{1-p_e}(t/T_2)^2+O[(t/T_2)^3]$, with its maximum at
$t_{\rm par}^{\star}/T_2=(1-p_e)\chi/p_e+O(\chi^{2}/p_e^{2})$ and
\begin{equation}
 \gamma_{\rm par}^{\star}
 =\frac{(1-p_e)\,\chi^{2}}{2p_e}+O(\chi^{3}/p_e^{2}).
 \label{sm:terminalopt}
\end{equation}
This expansion requires $\chi\ll p_e$.  Along $p_e=c\chi$, instead,
$\gamma_{\rm par}^{\star}/\chi\to
(1+c)e^{-1/(1+c)}-c$.  At $p_e=0$,
$\gamma_{\rm par}(t)=\gamma_{\rm q}(t)$ at every exposure, and their common
optimized value is
$\chi(1+\chi)^{-1-1/\chi}=\chi/e+O(\chi^2)$.  The quadratic law is therefore
a finite-temperature statement.

If $\Gamma_\phi=0$ and $p_e<1/2$, Eq.~\eqref{sm:nojump} has a single Kraus
operator and the conditional state is pure.  Put
$y=\sqrt{s_1/s_0}=e^{-(1-2p_e)t/(2T_1)}$.  Its Bloch coordinates are
\begin{equation}
 r_x=\frac{2y}{1+y^2},\qquad r_z=\frac{1-y^2}{1+y^2}.
 \label{sm:pureconditional}
\end{equation}
The Hadamard direction requires $2y=1-y^2$, hence $y=\sqrt2-1$ and
Eq.~\eqref{em:hadamard}.  At $p_e=1/2$ the two components never become equal
at finite exposure.  For $p_e=1/4$, the quiet-branch success probability is
$3\sqrt2-4\simeq0.243$.  Along the Hadamard-exposure curve $t=t_H(p_e)$,
the record-averaged channel reaches the complete stabilizer-preservation
boundary at the unique root of
$(1+\sqrt2)^{-1/(1-2p_c)}=p_c/(1-p_c)$, namely
$p_c\simeq0.192565$.

\section{One-use stabilizer-probe optimum}
\label{sm:probe}

Consider one use of a single-qubit channel, an arbitrary stabilizer input on the exposed qubit $A$ and a register $R$, followed by any stabilizer instrument.  Convexity reduces the input to a pure stabilizer state.  The reduced state $\rho_A$ of a pure stabilizer state is either a pure Pauli eigenstate or $I/2$.  In the first case the state factors across $A|R$.  In the second, the stabilizer normal form gives local Clifford operations on $A$ and $R$ such that
\begin{equation}
 (U_A\otimes U_R)\ket\Psi_{AR}
 =\ket{\Phi^+}_{A\bar R}\otimes\ket\phi_{R'},
 \label{sm:stabilizernormalform}
\end{equation}
where $\bar R$ is one logical qubit and $\ket\phi$ is a stabilizer
state~\cite{Fattal2004}.

For the phase-covariant channel, the largest product-input robustness surplus is $[\eta+k-1]_+$, attained by an equatorial Pauli eigenstate whenever it is positive.  For the one-ebit class, the identity
$(U_A\otimes I)\ket{\Phi^+}=(I\otimes U_A^T)\ket{\Phi^+}$ moves the exposed-rail Clifford to the reference.  Hence all rank-two inputs produce a Clifford image of the Choi state with stabilizer spectators.  Robustness is strongly monotone under the subsequent stabilizer instrument, so its average surplus cannot exceed $\R(J_\Phi)-1$~\cite{HowardCampbell2017}.

For the Choi state in Eq.~\eqref{sm:choimatrix}, the $ZZ$-parity measurement is nondisturbing because the state is already block diagonal in its parity sectors.  The even sector occurs with probability $p_+=(1+b)/2$ and decodes to Eq.~\eqref{sm:terminalstate}.  The odd sector has probability $p_-=(1-b)/2$ and is a computational-basis mixture.  A Clifford decoder writes the parity outcome to a computational-basis stabilizer ancilla.  Robustness is affine across these orthogonal classical sectors.  Strong monotonicity under measuring the ancilla gives one inequality.  The reverse follows by embedding optimal logical stabilizer decompositions in each sector and undoing the decoder.  Therefore
\begin{equation}
 \R(J_\Phi)
 =p_+\max\!\left\{1,\frac{2\eta+k}{1+b}\right\}+p_-
 =1+[\gamma_{\rm par}]_+.
 \label{sm:sectorrobustness}
\end{equation}
The same normalized-Choi robustness controls the cost of stabilizer propagation for postselected channels~\cite{Rall2019}.
Complete positivity implies $k\leq1-b$, and hence
$\gamma_{\rm par}-(\eta+k-1)=(1-b-k)/2\geq0$.  The ebit class is therefore never worse
than the product class.  For this phase-covariant family, terminal parity
attains the maximum one-use average robustness surplus
$[\gamma_{\rm par}]_+$, so larger
stabilizer codes cannot improve either the threshold or this one-use figure
of merit.  Applied to a trace-nonincreasing branch, the same sector argument
recovers the yield $[\gamma_\ell]_+$ in Lemma~\ref{lem:branch}.

\section{Thermal calibration lines and breaking channels}
\label{sm:thermal}

The numerical proxies quoted in the Letter use parameter triples reported
together for individual transmons: Table~S1 of
Ref.~\cite{Elder2020} and the unmonitored entries of Table~I in
Ref.~\cite{Goldblatt2024}.  The $74\,\mu\mathrm{s}$ value in the first
reference was obtained from Ramsey free induction (see Acknowledgments).
For the second device, the reported values are $T_1=67.0\pm0.3\,\mu\mathrm{s}$, $T_{2E}=68\pm1\,\mu\mathrm{s}$, and $p_e=0.80\pm0.05\%$. Treating the Hahn-echo time as a proxy gives $\chi=0.0068$.  Since
$\partial\chi/\partial T_{2E}=(1-p_e)/T_1=0.0148\,\mu\mathrm{s}^{-1}$,
the quoted $\pm1\,\mu\mathrm{s}$ error alone leaves its sign unresolved.
These coherence times are operational proxies for the homogeneous coherence
of the unsplit idle in the model.  We identify the reported low thermal
occupation with $p_e$ within a two-level truncation.
For the quoted estimates, the reported $T_1$ is interpreted as the
population-relaxation time $(\Gamma_\downarrow+\Gamma_\uparrow)^{-1}$.
If it instead denotes the downward lifetime
$T_{1,\downarrow}=\Gamma_\downarrow^{-1}$, then
$\chi=T_2/T_{1,\downarrow}-1$, giving $0.4510$ and $0.0149$ for the two
devices, respectively; the latter remains comparable to the coherence-time
uncertainty.

The terminal-parity margin $\gamma_{\rm par}(t)$ in Eq.~\eqref{sm:endgain} vanishes at $t=0$ and obeys
\begin{equation}
 e^{t/T_2}\dot \gamma_{\rm par}(t)
 =-\frac1{T_2}+\frac{1-p_e}{T_1}\,
 e^{-t(1/T_1-1/T_2)}.
 \label{sm:Mderivative}
\end{equation}
For $T_2>T_1$ the right-hand side decreases strictly, so $\gamma_{\rm par}$ is
single-peaked.  Its initial slope is $\chi/T_2$.  If $T_2\leq T_1$, then
$\eta\leq b$ and hence $\gamma_{\rm par}\leq p_e(b-1)\leq0$.  This proves the
$\chi$ criterion globally.  The same argument for
Eq.~\eqref{sm:averagedgain} yields initial slope $\chi_{\varnothing}/T_2$ and proves the
record-averaged $\ket+$ line.

At $p_e\leq1/4$, replacing $p_e$ by $2p_e$ in Eq.~\eqref{sm:endgain} yields
\begin{equation}
 \left.\gamma_{\varnothing}(t)\right|_{p_e}
 =\left.\gamma_{\rm par}(t)\right|_{2p_e},
 \label{sm:doubling}
\end{equation}
and therefore $\chi_{\varnothing}(p_e)=\chi(2p_e)$.  The interval between the two readout thresholds is
\begin{equation}
 \frac{T_1}{1-p_e}<T_2\leq
 \min\!\left\{\frac{T_1}{1-2p_e},2T_1\right\}.
 \label{sm:layer}
\end{equation}

The full output ellipsoid of a phase-covariant channel lies in the octahedron
iff~\cite{Patra2024}
\begin{equation}
 k+\sqrt{2\eta^2+b^2}\leq1.
 \label{sm:mb}
\end{equation}
Partial transposition of Eq.~\eqref{sm:choimatrix} shows that the channel is entanglement breaking iff~\cite{HorodeckiShorRuskai2003,KhatriSharmaWilde2020,Peres1996,Horodecki1996}
\begin{equation}
 4\eta^2+k^2\leq(1-b)^2.
 \label{sm:eb}
\end{equation}
For $T_2=2T_1$, set $u=e^{-t/(2T_1)}$, so $(\eta,b,k)=(u,u^2,\tau(1-u^2))$.  At $p_e=1/4$, or $\tau=1/2$, the identity point $u=1$ is completely stabilizer preserving, and for $t>0$ complete stabilizer preservation holds if and only if $u\leq1/3$.  Equations~\eqref{sm:eb} and \eqref{sm:mb} imply $u\leq(\sqrt7-2)/\sqrt3$ and $u\leq\sqrt{2/\sqrt3-1}$, respectively.  Hence the CSP, EB, and MB conditions all hold for $0<u\leq1/3$.  On $1/3<u\leq(\sqrt7-2)/\sqrt3$, the channel is both magic breaking and entanglement breaking but
$\gamma_{\rm par}=(3u-1)(1-u)/4>0$.  Its Choi state is separable and
nonstabilizer.

At $p_e=1/4$, the End Matter measure-and-prepare condition
$\eta\leq p_e(1-b)$ reduces to the narrower window
$u\leq\sqrt5-2$.  At the Hadamard exposure
$t_H=4T_1\ln(1+\sqrt2)$, Eq.~\eqref{em:mp} becomes
\begin{equation}
 \Phi_{t_H}=(3-2\sqrt2)(\Delta_X+\Delta_Y)
 +(17-12\sqrt2)\Delta_Z
 +(11\sqrt2-15)\mathcal P_{\to0}
 +(5\sqrt2-7)\mathcal P_{\to1}.
 \label{sm:mpdecomposition}
\end{equation}
All five weights are strictly positive and sum to one.  The three margins at
this point are $\gamma_{\rm q}=10-7\sqrt2$,
$\gamma_{\rm par}=7\sqrt2-10$, and $\gamma_{\varnothing}=4\sqrt2-6$.

Figure~\ref{sm:windows} compares the channel thresholds and the optimized
margins near the boundary.

\begin{figure}[t]
\centering
\includegraphics[width=0.9\textwidth]{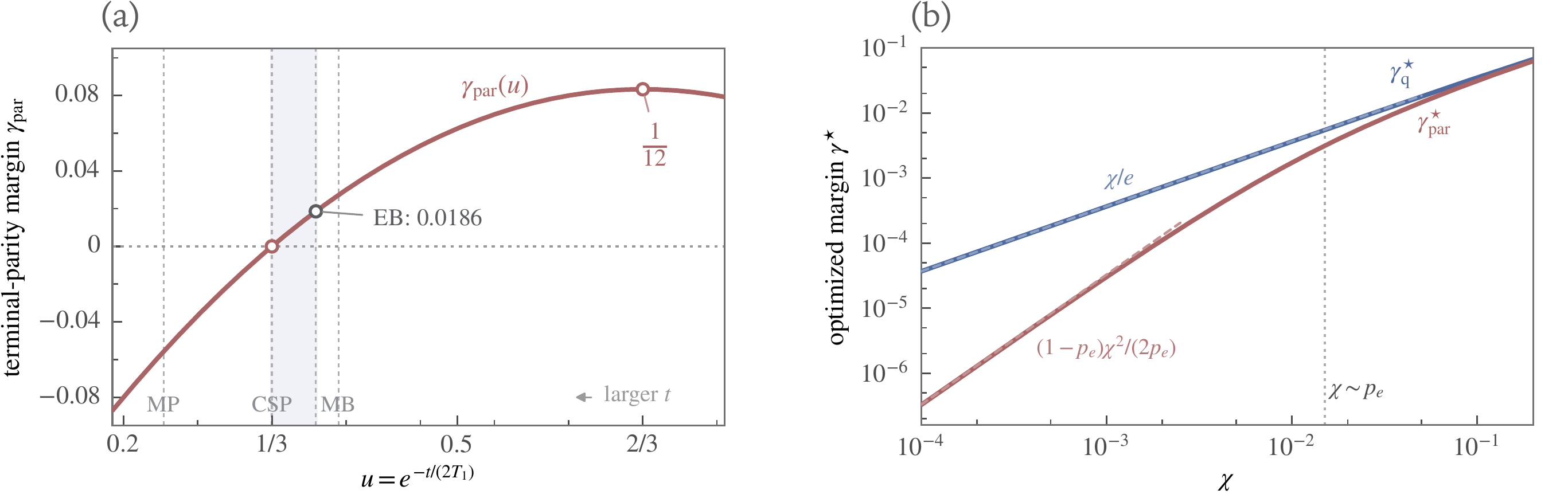}
\caption{Channel thresholds and optimized margins.
(a) At $p_e=1/4$ and $T_2=2T_1$, the horizontal axis is
$u=e^{-t/(2T_1)}$. The vertical guides mark the stabilizer
measure-and-prepare (MP), complete-stabilizer-preservation (CSP),
entanglement-breaking (EB), and magic-breaking (MB) thresholds. For
$0<u<1$, each condition holds to the left of its guide. The shaded interval
$1/3<u\leq u_{\rm EB}$ is entanglement breaking although its Choi state
is nonstabilizer. The terminal-parity margin takes the value $0.0186$
at $u_{\rm EB}$ and reaches its maximum $1/12$ at $u=2/3$.
(b) Exact optimized margins versus $\chi$ at $p_e=1.5\%$, with
$T_2/T_1=(1+\chi)/(1-p_e)$ varied. The dashed guides show the small-$\chi$
asymptotes $\chi/e$ and $\frac{(1-p_e)\chi^2}{2p_e}$. The quadratic law applies
only for $\chi\ll p_e$, with crossover on the scale $\chi\sim p_e$.}
\label{sm:windows}
\end{figure}

\section{Robustness, echo controls, and imperfect records}
\label{sm:scope}

For an $xz$-plane qubit, the robustness of magic in the convention $\R=1$ on the stabilizer polytope reduces to
$\R(\rho)=\max\{1,|r_x|+|r_z|\}$.  For each accepted state in Eqs.~\eqref{sm:averagedgain}, \eqref{sm:endgain}, and
\eqref{sm:exchangegain}, the signed margin obeys
$[\gamma]_+=p_{\rm s}[\R(\rho_{\rm s})-1]$.  Its positive part is the
success-weighted robustness surplus used throughout.  On the simulable side,
$-\gamma$ is a nonnegative coefficient in the corresponding stabilizer
decomposition.  Standard $\{I,H\}$ twirling
maps any positive-quadrant violation to Hadamard-axis polarization above
$1/\sqrt2$, after which Ref.~\cite{Reichardt2005} applies.

Within the time-homogeneous Markov model and a fixed calibrated phase frame,
an ideal instantaneous midpoint $X$ pulse provides a direct control for both
readout constructions.  Let $u=e^{-t/(2T_1)}$, with the final $X$ absorbed into the Pauli frame.  The transverse factor remains $u e^{-\Gamma_\phi t}$, while the magnitude of the displacement becomes $\tau(1-u)^2$.  The echoed terminal-parity margin therefore obeys
\begin{equation}
 2\gamma_{\rm par}^{\rm echo}(t)=-(1-\tau)(1-u)^2
 -2u\bigl(1-e^{-\Gamma_\phi t}\bigr)\leq0.
 \label{sm:echo}
\end{equation}
At zero temperature with $T_2=2T_1$, the terminal-parity echo lies exactly on
the facet.  For $t>0$, any finite temperature or homogeneous dephasing moves
it strictly inside.  The pulse removes the directed accumulation responsible for the terminal-parity window.

The pulse also makes the quiet-branch margin nonpositive.  In the quiet branch,
the two half-interval survival filters give the same diagonal weight
$u=e^{-t/(2T_1)}$, while the transverse weight is
$\eta=u e^{-\Gamma_\phi t}$.  The echoed quiet-branch margin is therefore
\begin{equation}
 \gamma_{\rm q}^{\rm echo}(t)
 =\eta-u
 =u\bigl(e^{-\Gamma_\phi t}-1\bigr)\leq0.
 \label{sm:quietecho}
\end{equation}
It lies on the stabilizer facet when $\Gamma_\phi=0$ and, for $t>0$, strictly
inside it when $\Gamma_\phi>0$.  Monitoring the dephasing bath as well would
define a different, unraveling-dependent instrument, which is not considered
here.

The ideal record in the Letter detects both upward and downward exchanges.  Let
$m_\uparrow$ and $m_\downarrow$ be the probabilities of missing the respective
jumps.  For the no-click output of $\ket+$, write
$n_0=\widetilde\rho_{00}$ and $n_1=\widetilde\rho_{11}$, so
$n_0(0)=n_1(0)=1/2$ and $2\widetilde\rho_{01}=\eta$.  The substochastic
population dynamics is
\begin{align}
 \dot n_0&=-\Gamma_\uparrow n_0
 +m_\downarrow\Gamma_\downarrow n_1,
 \nonumber\\
 \dot n_1&=m_\uparrow\Gamma_\uparrow n_0
 -\Gamma_\downarrow n_1.
 \label{sm:asymmetricpop}
\end{align}
The margin is
$\gamma_{m_\uparrow,m_\downarrow}=\eta-2\min\{n_0,n_1\}$.  Its right derivative at
the origin is
\begin{equation}
 \left.\dot\gamma_{m_\uparrow,m_\downarrow}\right|_{0^+}
 =-\frac1{T_2}+\max\!\left\{
 \begin{array}{l}
 \Gamma_\downarrow-m_\uparrow\Gamma_\uparrow,\\[-0.5mm]
 \Gamma_\uparrow-m_\downarrow\Gamma_\downarrow
 \end{array}\right\}.
 \label{sm:missedderivative}
\end{equation}
Complete positivity and $p_e\leq1/2$ imply
$1/T_2\geq1/(2T_1)\geq\Gamma_\uparrow$.  The second candidate for the derivative is at
most $-1/T_2+\Gamma_\uparrow\leq0$.  A positive initial slope is possible if
and only if
\begin{equation}
 (1-m_\uparrow)\chi+m_\uparrow\chi_{\varnothing}
 =\left.\chi\right|_{p_e\to(1+m_\uparrow)p_e}>0.
 \label{sm:missedline}
\end{equation}
Whenever this condition holds,
$\dot\gamma(0^+)=[(1-m_\uparrow)\chi+m_\uparrow\chi_{\varnothing}]/T_2$.  A missed
absorption enters the short-time test like additional excited-state
population.  In the strip $\chi_{\varnothing}\leq0<\chi$, the slope is positive for
$m_\uparrow<\chi/(\chi-\chi_{\varnothing})=(1-p_e)\chi/[p_e(1+\chi)]$.
Along the symmetric line, $m_\uparrow=m_\downarrow=1$ fully erases the record,
and the record-averaged $\ket+$ margin is already positive when $\chi_{\varnothing}>0$.

The generator in Eq.~\eqref{sm:asymmetricpop} is Metzler, so the margin is
nonincreasing in either miss probability at every exposure.  At fixed
$m_\downarrow$ in the strip $\chi_{\varnothing}\leq0<\chi$, the critical absorption-miss
probability for a positive-margin exposure is therefore no smaller than
this short-time value.

The terminal-parity measurement changes the asymmetric conclusion.  It rejects
every history containing only one missed exchange, so accepted contamination
requires a complete thermal round trip with both events missed.  For
independent misses the result depends only on the product of the two miss
probabilities.  Write
$x=t/T_1$ and
\begin{equation}
 M=m_\uparrow m_\downarrow,
 \qquad \Omega_M=\sqrt{\tau^2+(1-\tau^2)M},\qquad \tau=1-2p_e.
 \label{sm:monitoringparameters}
\end{equation}
The matrix exponential of Eq.~\eqref{sm:asymmetricpop}, now started
from a definite energy eigenstate, yields the no-click survival weights
\begin{align}
 s_0^{(M)}(x)&=e^{-x/2}\left[
 \cosh\frac{\Omega_Mx}{2}
 +\frac{\tau}{\Omega_M}\sinh\frac{\Omega_Mx}{2}\right],\nonumber\\
 s_1^{(M)}(x)&=e^{-x/2}\left[
 \cosh\frac{\Omega_Mx}{2}
 -\frac{\tau}{\Omega_M}\sinh\frac{\Omega_Mx}{2}\right].
 \label{sm:dm}
\end{align}
When
$\tau=M=0$, the products
$(\tau/\Omega_M)\sinh(\Omega_Mx/2)$ are defined by continuity as zero, so
$s_0^{(0)}=s_1^{(0)}=e^{-x/2}$.
Applying the no-click map to one rail of a stabilizer ebit and retaining the
even terminal-parity outcome produces the unnormalized logical state
\begin{equation}
 \sigma_{\rm nc,+}=\frac12
 \begin{pmatrix}s_0^{(M)}&\eta\\ \eta&s_1^{(M)}\end{pmatrix},
 \qquad
 \gamma^{(M)}(t)=\eta-s_1^{(M)}(t/T_1).
 \label{sm:directionterminal}
\end{equation}
Parenthesized superscripts label the monitoring interpolation, not powers.
For a noninverted bath, $s_0^{(M)}\geq s_1^{(M)}$.  The same
$s_0^{(M)}$ and $s_1^{(M)}$ are the $A,D$ coefficients of the
no-click map before terminal parity.  Lemma~\ref{lem:branch} therefore makes
$\eta\leq s_1^{(M)}$ the map-level stabilizer criterion.  By contrast,
Eq.~\eqref{sm:missedline} is the short-time line of the product probe
$\ket+$.  Perfect monitoring and complete erasure give, respectively,
\begin{equation}
 s_1^{(0)}(x)=e^{-(1-p_e)x}=s_1,
 \qquad
 s_1^{(1)}(x)=p_e+(1-p_e)e^{-x}.
 \label{sm:dmlimits}
\end{equation}
The trajectory expansion is a power series in $M$ with nonnegative
coefficients and a nonzero one-round-trip term, so $s_1^{(M)}$ is strictly
increasing for $0<p_e<1/2$ and $t>0$.

In the window $\gamma_{\rm par}(t)\leq0<\gamma_{\rm q}(t)$, the endpoint
identities $\gamma^{(0)}(t)=\gamma_{\rm q}(t)$ and
$\gamma^{(1)}(t)=\gamma_{\rm par}(t)$ show that the equation
$\eta=s_1^{(M_c)}(t/T_1)$ uniquely defines the threshold
$M_c(t)\in[0,1]$.  At
$p_e=1/4$ and $T_2=2T_1$, it is the dashed curve $M_c(u)$ in
Fig.~\ref{fig:ladder}(c), with its endpoints taken by continuity.
Equation~\eqref{sm:dm} interpolates between the terminal-parity and fully
resolved readouts.  If either direction is monitored perfectly, then $M=0$
and Eq.~\eqref{sm:directionterminal} reproduces the accepted state, probability,
and margin of the fully resolved quiet branch.  A no-click history can then
contain at most one jump in the unmonitored direction, which the even
terminal-parity outcome rejects.

The expansion of Eq.~\eqref{sm:dm} is
\begin{equation}
 s_1^{(M)}(x)=1-(1-p_e)x
 +\frac{1-p_e}{2}\bigl(1-p_e+p_eM\bigr)x^2+O(x^3).
 \label{sm:dmexpansion}
\end{equation}
Hence $\dot\gamma^{(M)}(0^+)=\chi/T_2$ for every $M$.  Since
$\gamma^{(M)}(0)=0$, $\chi>0$ gives positive margin at sufficiently
short exposure for every $M\in[0,1]$.  Conversely,
$s_1^{(M)}\geq s_1^{(0)}=s_1$, so $\chi\leq0$ implies
$\eta\leq s_1^{(M)}$ at every exposure.  Detector loss therefore leaves
the selectable-exposure existence boundary at $\chi=0$.

At $\chi=0$ and $p_e>0$,
$s_1^{(M)}>s_1^{(0)}=\eta$ for every $M>0$ and $t>0$, with the departure
beginning at order $t^2$.  At $p_e=0$, no thermal round trip exists,
$s_1^{(M)}=e^{-x}$ for every $M$, and the branch remains on the facet.
At $p_e=1/4$, $T_2=2T_1$, and
$t=4T_1\ln(1+\sqrt2)$, the fixed-time threshold satisfies
\begin{equation}
 \tanh\!\left[\Omega_{M_c}\ln(1+\sqrt2)\right]
 =\frac{1}{2\Omega_{M_c}},
 \qquad M_c\simeq0.54835.
 \label{sm:fixedmonitoringthreshold}
\end{equation}
The branch has positive margin for $M<M_c$.  Equal detectors therefore need
per-direction efficiency above $1-\sqrt{M_c}\simeq0.2595$.  At efficiency
$0.5$, where $M=1/4$, the margin from Eq.~\eqref{sm:directionterminal} is
$0.05744$, or $57\%$ of the ideal value $10-7\sqrt2\simeq0.10051$.

These formulas also provide a two-sided finite-time criterion.  If
$\eta\leq s_1^{(M)}$, the accepted no-click, even-parity map has the
nonnegative decomposition
\begin{equation}
 \mathcal E_{\rm nc,+}=\eta\mathcal I+
 (s_0^{(M)}-\eta)\mathcal P_{0\to0}
 +(s_1^{(M)}-\eta)\mathcal P_{1\to1}.
 \label{sm:directiondecomposition}
\end{equation}
All branches containing a detected click are measure-and-reset operations, and
a no-click odd branch is diagonal.  Hence the fixed-duration instrument has a
nonnegative stabilizer decomposition when $\eta\leq s_1^{(M)}$.  When
$\eta>s_1^{(M)}$, Eq.~\eqref{sm:directionterminal} is outside the
octahedron and is distillable.

Independent dark clicks at rate $\Gamma_{\rm d}$ multiply the accepted no-click map by $e^{-\Gamma_{\rm d}t}$.  They leave its normalized state and margin sign unchanged, while the success probability and success-weighted margin acquire the same factor.  For an independent, state-independent terminal-readout model, let $r_{\rm FR}$ and $r_{\rm FA}$ be the false-rejection and false-acceptance probabilities.  With no further parity check, the nominally even branch is
$(1-r_{\rm FR})\mathcal E_{\rm even}+r_{\rm FA}\mathcal E_{\rm odd}$.
For the Bell-pair probe followed by the even-sector decoder, its success-weighted margin in a noninverted bath is
\begin{equation}
 (1-r_{\rm FR})\gamma_{\rm par}-r_{\rm FA}p_e(1-b).
\end{equation}
False rejection alone rescales the success probability and margin, while false acceptance shifts the short-time condition to $(1-r_{\rm FR})\chi>r_{\rm FA}(\chi-\chi_{\varnothing})$.  The false-acceptance term
vanishes at $p_e=0$.

\end{document}